\documentclass[preprint]{elsarticle}
\usepackage{natbib,geometry,amssymb,amsmath,amsthm,txfonts,hyperref}
\usepackage{graphicx,color}
\journal{Mathematical Biosciences}

\newtheorem{thm}{Proposition}

\newdefinition{rmk}{Remark}
\newtheorem{proposition}[thm]{Proposition}

\begin{document}
\begin{frontmatter}
\title{Disease Persistence in Epidemiological Models: \\ The Interplay between Vaccination and Migration}

\author[label1]{Jackson Burton}
\author[label1]{Lora Billings \corref{cor1}}
\cortext[cor1]{Corresponding author}
\ead{billingsl@mail.montclair.edu}
\author[label2]{Derek A.~T.~Cummings}
\author[label3]{Ira B.~Schwartz} 
\address[label1]{Montclair State University, Department of Mathematical Sciences, Montclair, NJ 07043}
\address[label2]{Johns Hopkins Bloomberg School of Public Health, Department of International Health, Baltimore, MD, 21205}
\address[label3]{US Naval Research Laboratory, Code 6792, Nonlinear System Dynamics Section, Plasma Physics Division, Washington, DC 20375 }

\begin{abstract}
We consider the interplay of vaccination and migration rates on disease persistence in epidemiological systems. We show that short-term and long-term migration can inhibit disease persistence. As a result, we show how migration changes how vaccination rates should be chosen to maintain herd immunity. In a system of coupled SIR models, we analyze how disease eradication depends explicitly on vaccine distribution and migration connectivity. The analysis suggests potentially novel vaccination policies that underscore the importance of optimal placement of finite resources.
\end{abstract}
		
\begin{keyword}
epidemics, migration, vaccination, herd immunity
\end{keyword}

\end{frontmatter}

\section{Introduction}
Countries are increasingly connected by travel and economics. Due to economic disparities and political turmoil, extreme heterogeneity exists in childhood vaccination coverage across the two sides of multiple national boundaries. It has been suggested that the immunization coverage of neighboring countries or those countries well connected by travel can or should be used when crafting national level immunization policy. In the case of hepatitis B, Gay and Edmunds \cite{Gay} argue that it would be four times more cost effective for the United Kingdom to sponsor a vaccination program in Bangladesh than to introduce its own universal program. When indigenous wild poliovirus was eradicated in all but four endemic countries in 2005: India, Nigeria, Pakistan and Afghanistan, it was exported from northern Nigeria and northern India and subsequently caused $>50$ outbreaks and paralyzed $>1500$ children in previously polio-free countries across Asia and Africa \cite{Aylward}. And in 2007, the WHO estimated that there were 197,000 measles deaths, despite the 82\% worldwide vaccination coverage. In countries where measles has been largely eliminated, cases imported from other countries remain an important source of infection \cite{WHO3}. It is clear that a country needs to be concerned with the vaccination rate of a neighboring country as well as its own. 

On another scale, vaccination policies must also take into consideration the subpopulation dynamics within a country. Wilson, et al. \cite{blower} models linked urban and rural epidemics of HIV and discusses how to optimize a limited treatment supply to minimize new infections. Cummings et al. \cite{cummings} uses data to identify a distinct pattern in the periodicity of measles outbreaks in Cameroon before the widespread vaccination efforts of the Measles Initiative. The southern part of Cameroon experienced a significant measles epidemic approximately every three years. In contrast, the three northern provinces contend with annual measles epidemics. In 2000 and 2001, these cyclic outbreaks coincided, exacerbating the situation and causing a much more severe epidemic \cite{cummings}.

Noting that a small contribution of infections from one population to another could drive a new type of epidemic that would not normally occur, we study how migration between populations could change dynamics  and respective herd immunity levels in metapopulation models. We analyze a model of a disease imported between subpopulations of a region by short-term and long-term migration with limited vaccination coverage. Our initial study is based on the analysis of a system of canonical SIR compartmental models. The system allows the rigorous proof of the qualitative affects migration has on herd immunity. The model can be enhanced to include more compartments or seasonal forcing, but most of these systems will require numerical exploration of trends in spatial synchrony and bifurcation analysis, which will be explored in future papers. In this article, we revisit the fundamental ways migration is modeled in metapopulation models and how it fundamentally affects  herd immunity. 

Migration is often treated as a phenomenological input to maintain incidence in a population that might experience local fade-out \cite{ferguson_ADE}.  Long-term migration has been analyzed by Liebovitch and Schwartz \cite{liebovitch}, with a thorough derivation of the linear flux term coupling the patches. This approach also agrees with the classes of models proposed by Sattenspiel and Dietz \cite{Sattenspiel} and Lloyd and Jansen \cite{Lloyd2004}. Keeling and Rohani \cite{Keeling02} investigated the spatial coupling of dynamics exhibited in models using multiple formulations of migration including mass-action coupling and linear flux terms. However, they did not explore the impact of coupling in the presence of vaccination needed to maintain disease free states. Additional analysis of mixed long-term and short-term migration in transport-related disease spread can be found in \cite{Cui2006,Cui2007,Liu2009}. These papers derive the global asymptotic stability of the disease free state for a new disease. Because there is no vaccination, the papers conclude that it is essential to strengthen restrictions of passenger travel as soon as the infectious diseases appear. 

Our paper considers how migration directly affects the vaccination levels needed for herd immunity against a known disease and how that would impact optimum usage of limited vaccination supplies. We investigate the dynamics of models that include mass-action coupling, an assumption that assumes mixing occurs at fast time scales, and linear migration, which is more consistent with mixing occurring at long time scales.  The organization of this paper is as follows:  We introduce a coupled compartmental model in Section \ref{Sec:model} and perform stability analysis of the disease free state as a function of the migration and vaccination rates.  We also consider normal forms of the bifurcations created by the short-term and long-term migration dynamics. Section \ref{Sec:vaccination} describes how vaccination rates should be adjusted with respect to short-term and long-term migration levels to preserve herd immunity. Section \ref{Sec:conclusion} has a summary of our observations and conclusions.

		\section{The Model} \label{Sec:model}

We start with the classic Susceptible, Infected, Recovered (SIR) model. Let $S$, $I$, and $R$ denote the number of people in each of the disease classes for a population of size $N$. Let the parameters $\beta > 0$ denote the contact rate, $\mu > 0$ denote the birth/death rate, and $\kappa > 0$ denote the recovery rate. The vaccination rate, $0 \le v \le 1$, represents the removal of a percentage of the incoming newborn population to recovered.  The standard form for this system is 
                 \begin{eqnarray}
		\frac{dS}{dt} &=& (1-v) \, \mu N - \frac{\beta S I}{N} - \mu S, \nonumber \\ 
		\frac{dI}{dt} &=& \frac{\beta SI}{N}  -\kappa I -\mu I,   \label{eq:SIR_classic} \\  
		\frac{dR}{dt} &=& v \mu N + \kappa I -\mu R. \nonumber
		\end{eqnarray} 
The death rates in the classes balance the births so that the population size $N > 0$ is constant. For a detailed analysis of the single patch formulation of this system, see Hethcote \cite{Hethcote}. 

We now consider two coupled subpopulations where the disease dynamics of each population are described by the SIR model. Let $S_k$, $I_k$, and $R_k$ denote the number of people in each of the disease classes, $\mu_k > 0$ denote the birth/death rates, and $v_k$ denote the vaccination rates of subpopulations $N_k$ for $k=1,2$.  To model long-term movement (linear mixing), let $c_1 \ge 0$ denote the rate of migration from population two to population one and vice versa for the rate $c_2 \ge 0$.  To model short-term movement (mass action mixing), let $0 \le c_3 \le 1 $ be a scaling of the number of infectives from one population who move into the other population for a short time and mix with the susceptibles to produce additional infections.  Because $\beta$ is proportional to the average number of contacts a person can make per unit time, we distribute the contacts for the susceptibles between the infected people by mass action within and outside the population by using the prefactors $c_3$ and $(1-c_3)$ respectively as in Keeling and Rohani \cite{Keeling02}. The coupled two population model is as follows:
		\begin{eqnarray}
		\frac{dS_1}{dt} &=& (1-v_1) \, \mu_1 N_1-\frac{(1-c_3) \,\beta \, S_1I_1}{N_1}-\frac{c_3\,\beta \, S_1 I_2}{N_1}-\mu_1S_1+c_1S_2-c_2 S_1, \nonumber \\ 
		\frac{dI_1}{dt} &=& \frac{(1-c_3)\,\beta \, S_1I_1}{N_1}+\frac{c_3\,\beta \, S_1 I_2}{N_1}  -\kappa I_1 -\mu_1 I_1 +c_1I_2-c_2 I_1, \nonumber \\  
		\frac{dR_1}{dt} &=& v_1\mu_1N_1 + \kappa I_1 -\mu_1R_1+c_1R_2-c_2 R_1, \label{eq:SIR_coupled} \\  
		\frac{dS_2}{dt} &=&  (1-v_2) \, \mu_2 N_2-\frac{(1-c_3)\,\beta \, S_2I_2}{  N_2}-\frac{c_3\,\beta \, S_2 I_1}{ N_2}-\mu_2S_2+c_2 S_1-c_1S_2, \nonumber \\  
		\frac{dI_2}{dt} &=& \frac{(1-c_3)\,\beta \,S_2I_2}{ N_2}+\frac{c_3\,\beta \, S_2 I_1}{ N_2}  -\kappa I_2 -\mu_2 I_2 +c_2 I_1-c_1 I_2, \nonumber \\  
		\frac{dR_2}{dt} &=& v_2\mu_2  N_2 + \kappa I_2 -\mu_2R_2+c_2 R_1-c_1 R_2 .   
		\nonumber
		\end{eqnarray} 

We keep the number of people in the subpopulations constant by letting $\rho = N_2/N_1$ and setting the constraint $c_2 = c_1 \rho$.  This system is overdetermined by the subpopulation constraints, $S_k+I_k+R_k=N_k$ for $k=1,2$, and therefore the analysis omits the variables $R_k$ for $k=1,2$. 

Motivated by the distinct subpopulation dynamics of Cameroon described in Cummings et al. \cite{cummings}, numerical simulations will use parameters based on Cameroon demographics. The values are listed in Table~\ref{table:full_params}.   The subpopulation sizes are totals for the northern and southern regions based on data in \cite{cummings}.  The birth/death rates are averages over the northern and southern regions based on data in \cite{cummings}. The recovery rate is a  parameter that is derived from the biological characteristics of measles.  The contact rate was estimated using the average age of incident measles cases over the period 1998-2006 from passive surveillance data \cite{cummings}.   The specific results here are fairly insensitive to changes to $\beta$.  An SIR model is used here without the exposed class but we expect the inclusion of an exposed class would not substantively change our qualitative results.

\begin{table}
		\caption{Parameter Values for Model Based on Cameroon Data} 
		\centering 
		\begin{tabular}{c | c | c | c} 
		\hline\hline 
		Parameter & Value & Unit & Description\\  
		\hline 
		$N_1$ & 4,451,000 & people & Northern subpopulation size \\ 
		$N_2$ & 10,212,000 & people & Southern subpopulation size \\ 
		$\rho$ & $2.2943$ & none & Ratio of $N_2/N_1$   \\
		$\beta$ & $700 $ & year$^{-1}$ & Contact rate \\
		$\kappa$ & $100$ & year$^{-1}$ & Measles recovery rate  \\
		$\mu_1$ & $.0428$ & year$^{-1}$ & Birth and death rate for $N_1$  \\
		$\mu_2$ & $.0329$ & year$^{-1}$ & Birth and death rate for $N_2$ \\
		\end{tabular}
		\label{table:full_params} 
		\end{table}
		
\subsection{General system analysis}
We start with a general analysis of the system to determine the conditions necessary for the populations to be disease free.		
\begin{thm} System (\ref{eq:SIR_coupled}) has a disease free equilibrium (DFE) and is given by 
\begin{equation}
 (S_1,I_1,S_2,I_2)=(N_1 \hat{S}_1,0,N_2 \hat{S}_2,0 ), 
\end{equation} 
for 
\begin{eqnarray}
		\hat{S}_1 = \frac{(1-v_1)(\mu_1 c_1 + \mu_1 \mu_2) + (1-v_2) \mu_2 c_1 \rho }{  \mu_1 c_1 + \mu_1 \mu_2 + \mu_2 c_1 \rho}, \\
 \hat{S}_2 = \frac{(1-v_1) \mu_1 c_1+(1-v_2)(\mu_1 \mu_2 + \mu_2 c_1 \rho ) }{  \mu_1 c_1 + \mu_1 \mu_2 + \mu_2 c_1 \rho}. 
		\end{eqnarray}
\end{thm}

Note that the DFE does not depend on the short-term migration parameter $c_3$. Without long-term migration ($c_1=0$), the DFE simplifies to $(S_1,I_1,S_2,I_2)=\left(N_1(1-v_1),0,N_2(1-v_2),0 \right)$, which is the steady state for the uncoupled system. Also, there is no steady state for which the disease dies out in only one of the two subpopulations if $c_1>0$. 
 
The local stability of the DFE can be determined by the eigenvalues of the Jacobian of the system evaluated at the DFE. The resulting eigenvalues are
		\begin{eqnarray}
                  \lambda_1 &=& -\frac{1}{2}\left(c_1+c_1 \rho+\mu_1+\mu_2 +
                                              \sqrt{(c_1+c_1 \rho+\mu_1-\mu_2)^2-4c_1(\mu_1-\mu_2)}\right),\\
                  \lambda_2 &=& -\frac{1}{2}\left(c_1+c_1 \rho+\mu_1+\mu_2 -
                                              \sqrt{(c_1+c_1 \rho+\mu_1-\mu_2)^2-4c_1(\mu_1-\mu_2)}\right),\\
		\lambda_3 &=& -\frac{1}{2}\left(c_1+c_1 \rho+\mu_1+\mu_2+2\kappa - (1-c_3)\,\beta\,(\hat{S}_1+\hat{S}_2) \,  + \sqrt{W} \right), \\
		\lambda_4 &=& -\frac{1}{2}\left(c_1+c_1 \rho+\mu_1+\mu_2+2\kappa - (1-c_3)\,\beta\,(\hat{S}_1+\hat{S}_2) \,  - \sqrt{W} \right), \label{lambda4}
		\end{eqnarray}
		for
\begin{equation}
W = 4\,\beta\,c_3\,c_1 \left( \hat{S}_2+ \rho\,\hat{S}_1\right)  +4 \left(\beta^2 c_3^2 \hat{S}_1  \hat{S}_2+c_1^2 \rho \right) + \left( (1-c_3)\,\beta \left( \hat{S}_1- \hat{S}_2 \right) +c_1-c_1 \rho-\mu_1+\mu_2 \right) ^{2} .
\label{eq:dom_eig_full}
\end{equation}
The DFE is locally stable if the maximum value of the real parts of this set of eigenvalues is negative. 
\begin{thm} The eigenvalue $\lambda_4$ determines the local stability of the DFE. 
\end{thm}
\begin{proof}
We can show $\lambda_1$ and $\lambda_2$ are always negative. First, let  
\begin{equation}
		\begin{array}{rl}
		\theta_1&=c_1+c_1 \rho+\mu_1+\mu_2, \\[.1 in]
		 \theta_2&= (c_1+c_1 \rho+\mu_1-\mu_2)^2-4c_1(\mu_1-\mu_2).\\
		\end{array}
		\end{equation}
If we consider $\theta_2$ as a quadratic expression in $c_1$ with leading coefficient $(\rho+1)^2$, it attains an absolute minimum at $d\theta_2/dc_1=0$.  Solving this equation gives $c_1 = (\mu_1-\mu_2)/(\rho+1)^2$.  Substituting this expression into $\theta_2$ to find the absolute minimum gives 
$4\rho(\mu_1-\mu_2)^2/(\rho+1)^2 >0$. Therefore $\theta_2>0$ for all $c_1$, which implies $\lambda_1$ and $\lambda_2$ are real valued. 

Upon inspection we see that $\theta_1>0$, $\theta_1+\sqrt{\theta_2}>0$, and therefore $\lambda_1<0$. For $\lambda_2 <0$,  we require $\theta_1> \sqrt{\theta_2}$. This is equivalent to $c_1 > -\mu_1\mu_2/(\mu_1+ \mu_2\rho)$. Since we assume $c_1 >0$, this is always true. Therefore $\theta_1 - \sqrt{\theta_2}>0$, which implies $\lambda_2 <0$.  

For our parameter assumptions, we see that $W>0$ by inspection. This implies $\lambda_3$ and $\lambda_4$ are real valued and $\lambda_4>\lambda_3$. The only way for the DFE to be unstable is for $\lambda_3>0$ or $\lambda_4>0$.  Because of the ordering, a sign change would have to happen for $\lambda_4$ first. Therefore, $\lambda_4$ determines the stability of the DFE. 
\end{proof}

		\begin{figure}
		\begin{centering}
		\includegraphics[height=65mm]{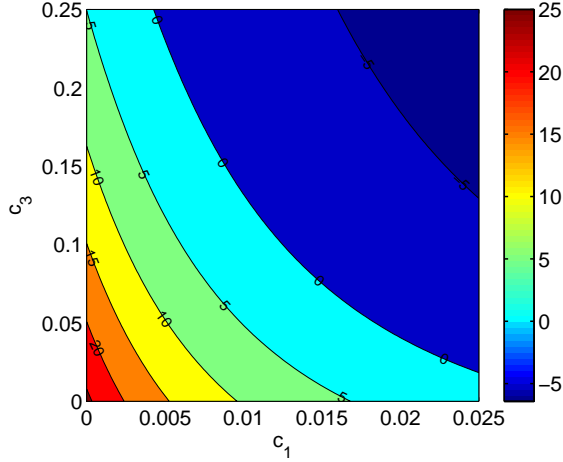}
\caption{ {\label{fig:c1c3} Contour plot of $\lambda_4$ values as a function of $c_1$ and $c_3$. Parameters are given by the values in Table~\ref{table:full_params}, with $v_1=0.82$ and $v_2=0.90$. The DFE is unstable for $\lambda_4>0$, which occurs for smaller values of $c_1$ and $c_3$.  }}
		\end{centering}
		\end{figure}

To quantify how each migration type effects the stability of the DFE, we can monitor the sign of $\lambda_4$ as we vary $c_1$ and $c_3$. As an example, we show a contour plot of $\lambda_4$ in Fig.~\ref{fig:c1c3} using the parameters in Table~\ref{table:full_params}, with $v_1=0.82$ and $v_2=0.90$. When $\lambda_4>0$, the DFE is unstable. From the figure, you can see that as $c_1$ and $c_3$ decrease, $\lambda_4$ increases and the DFE becomes unstable. We now explore the underlying conditions necessary in each subpopulation for which migration can cause the die out or invasion of a disease.

In the absence of migration,  we recover the basic reproductive numbers scaled by vaccination for
each subpopulation for the uncoupled system  \cite{Hethcote}. Specifically
when $c_1=c_3=0$, 
\begin{equation} \label{Eq:basicrep}
\hat{R}_1(v_1) = \frac{\beta(1-v_1)}{\kappa+\mu_1} ~~\mbox{and} ~~ \hat{R}_2(v_2)=\frac{\beta(1-v_2)}{\kappa+\mu_2}
\end{equation}
We omit the arguments for $\hat{R_k}$ for $k=1,2$ from here on, unless otherwise specified. The basic reproductive number is the quantity that defines the threshold between disease absence and persistence, and for the canonical SIR model without vaccination $R_0 = \frac{\beta}{\kappa+\mu}$. We write these expressions as a function of the vaccination rate in the subpopulation noting that the inequalities $\hat{R}_k<1$ for $k=1,2$ implies that the DFE in each subpopulation is locally stable. At $\hat{R}_k=1$ for $k=1,2$, these two expressions also represent transcritical bifurcations for the uncoupled system. As an example, for the parameters used in Fig.~\ref{fig:c1c3}, $\hat{R}_1(0.82)>1$ and the disease would be endemic in $N_1$. Conversely, $\hat{R}_2(0.90)<1$ and the disease would die out in $N_2$. This motivates us to examine the effect migration has on a simple system with a transcritical bifurcation in each component.  

\subsection{Normal Form for Linear Mixing}

To understand how long-term migration directly affects the stability of the
DFE, we rewrite the system as the normal form of a transcritical bifurcation
with linear coupling. Since the SIS model has the same topology near the DFE as
the SIR model, we consider the standard SIS model with births and deaths \cite{Murray} 
\begin{subequations}
\begin{align}
\frac{ds}{dt} &= \mu - \beta s i + \kappa i - \mu s,  \label{eq:SIS:S} \\
\frac{di}{dt} &= \beta s i - \kappa i - \mu i,  \label{eq:SIS:I} 
\end{align}
\end{subequations}
with nondimensional variables representing percentages of the population. Since
$s+i=1$, the system is overdetermined and we need only to solve $di/dt$. 

In the first subpopulation, let $x=i$ and $1-x=s$. Substitute these variables into Eq.~(\ref{eq:SIS:I}) and rescale time by $\beta$. We repeat the process using the variable $y$ to represent the second population. By adding the linear migration terms to these equations, we find
\begin{eqnarray}
\dot{x}&=&r_1x-x^2- \alpha x+ \alpha y,  \label{norm_form_coupling} \\
\dot{y}&=&r_2y-y^2-\alpha y+\alpha x.  \nonumber
\end{eqnarray}
Here, the bifurcation parameter $r_{k}=(\hat{R}_{k}-1)/\hat{R}_{k}$ for $k=1,2$ from Eq.~(\ref{Eq:basicrep}). In addition, we rescaled the long-term migration rate as the parameter $\alpha=c_1/\beta$.

The steady state $(x,y)=(0,0)$ is equivalent to the DFE in the full system in Eq.~(\ref{eq:SIR_coupled}). In the absence of coupling ($\alpha=0$), a transcritical bifurcation will occur in the $x$ system, transferring the stability from $x=0$ to $x=r_1$ at $r_1=0$. The dynamics are similar for $y$, respectively. Linearizing about the steady state $(x,y)=(0,0)$ yields two eigenvalues,
\begin{eqnarray} \label{norm_eigs}
\Lambda_{1}&=&\frac{1}{2}\left( r_1+r_2 -2\alpha +\sqrt{4\alpha^2 + (r_1-r_2)^2}\right), \\
\Lambda_{2}&=&\frac{1}{2}\left( r_1+r_2 -2\alpha -\sqrt{4\alpha^2 + (r_1-r_2)^2}\right). \nonumber
\end{eqnarray} 
The following analysis uses this linearization approach to conclude when long-term migration can change the stability of this steady state.

We start by considering the case of two isolated endemic populations, which have basic reproduction numbers greater than one. We ask if it is possible to stabilize the die out state through the coupling parameter, $\alpha$.  
\begin{thm} \label{thm:twobig}
If $r_1,r_2 > 0$, then the fixed point $(x,y)=(0,0)$ is unstable for all $\alpha  \in [0, \infty)$.
\end{thm}
\begin{proof} 
Upon inspection, $\Lambda_1$ is the dominant eigenvalue.  Assuming $r_1,r_2 > 0$, then $\Lambda_1>0 $ implies
\begin{equation}
(r_1+r_2) + \left(\sqrt{4\alpha^2 + (r_1-r_2)^2}\right) > 2 \alpha.
\end{equation}  
Squaring both sides and simplifying, we find
\begin{equation} \label{eq:lam1}
r_1^2+r_2^2+(r_1+r_2)\sqrt{4\alpha^2 + (r_1-r_2)^2} > 0, 
\end{equation}  
which is always true. Therefore, $(x,y)=(0,0)$ is unstable for all $\alpha  \in [0, \infty)$. 
\end{proof}
We can interpret this abstract result in the original system by concluding that for two isolated endemic populations, the amount of long-term migration is irrelevant to the persistence of the disease. The stability of the DFE cannot be changed by migration and intervention by vaccination is necessary for disease die out.   

Next, consider the case where we have sufficient vaccination so that one of the basic reproductive numbers is less than one, while the other is not. Again, we ask under what conditions the coupling parameter can stabilize the die out state.
\begin{thm}
Without loss of generality,  we assume $r_1>0$ and $r_2<0$. Case 1: If $-r_2<r_1$, then the fixed point $(x,y)=(0,0)$ is unstable for all $\alpha  \in [0, \infty)$. Case 2: If $-r_2>r_1$, then there exists some $\alpha^* \in [0,\infty)$ such that the fixed point $(0,0)$ is stable for all $\alpha >\alpha^*$. 
\end{thm} 
\begin{proof} In both cases, assume $r_1>0$ and $r_2<0$. 

Case 1: For $-r_2<r_1$, $\Lambda_1>0$ reduces to the relationship in Eq.~(\ref{eq:lam1}). This is always true for $0<r_1+r_2$ and we conclude $(x,y)=(0,0)$ is unstable for all $\alpha  \in [0, \infty)$.

Case 2: For $-r_2>r_1$, $\Lambda_1<0$ reduces to the relationship 
\begin{equation} \label{eq:neglam1}
\left(\sqrt{4\alpha^2 + (r_1-r_2)^2}\right) < 2\alpha - (r_1+r_2). 
\end{equation}  
Because both sides are positive, we can square both sides to find
\begin{equation} \label{Eq:minmig}
\alpha >\frac{r_1r_2}{r_1+r_2}>0 .
\end{equation}    
There exists an $\alpha^*=r_1r_2/(r_1+r_2)$ for $\alpha^* \in [0,\infty)$. Therefore, for $\alpha >\alpha^*$, it follows that $\Lambda_1<0$ and $(x,y)=(0,0)$ is stable. 
\end{proof}
This result implies that in a system with one population supporting an endemic state, there is a minimum amount of migration necessary for the system to achieve stability of the DFE. In fact, we can interpret the requirement of   $-r_2>r_1$ as $y=0$ in the uncoupled system is more stable than $x=0$. Therefore, $y$ is sharing its extra stability with $x$. 

For completeness, we can also show that long-term migration cannot change the stability of a stable die out state for two isolated populations that have basic reproduction numbers less than one.
\begin{thm}
If $r_1$, $r_2 < 0$, then the fixed point $(x,y)=(0,0)$ is stable for all $\alpha  \in [0, \infty)$.
\end{thm} 
\begin{proof}  For $r_1$, $r_2 < 0$, $\Lambda_1<0$ reduces to Eq.~(\ref{eq:neglam1}). Because both sides are positive, we can square both sides to find
\begin{equation} 
\alpha >\frac{r_1r_2}{r_1+r_2} .
\end{equation}    
Since $\frac{r_1r_2}{r_1+r_2}<0$, $\Lambda_1<0$ for all $\alpha \in [0,\infty)$ and $(x,y)=(0,0)$ is stable. 
\end{proof}
We conclude that long-term migration has a positive effect on the stability of the DFE. The mixing in all classes diffuses the force of infection, making it harder for the disease to persist. In applications where migration is common, this effect might be significant.

\subsection{Normal Form for Mass Action Mixing}

To capture the effect of short-term migration for each subpopulation in Eq.~(\ref{eq:SIR_coupled}), we follow the construction of model for linear mixing by substituting $x$ and $y$ for $i$ in Eq.~(\ref{eq:SIS:I}). In this system, we use mass action coupling with the parameter $\sigma=c_3$ controlling the mass action mixing strength. Specifically, the term $\sigma(1-x)y$ represents the infectious person from $y$ coming into contact with a susceptible from $x$. The system take the form 
\begin{eqnarray}
\dot{x}&=& \left(\frac{r_1-\sigma}{1-\sigma}\right)x-x^2-\frac{\sigma}{1-\sigma}(1-x)y,\\
\dot{y}&=& \left(\frac{r_2-\sigma}{1-\sigma}\right)y-y^2-\frac{\sigma}{1-\sigma}(1-y)x.  \nonumber
\label{norm_form_masscoupling}
\end{eqnarray}
Again, the bifurcation parameter $r_{k}=(\hat{R}_{k}-1)/\hat{R}_{k}$ for $k=1,2$ from Eq.~(\ref{Eq:basicrep})  and time has been rescaled by $\beta(1-\sigma)$. 

Performing the linearization about the steady state $(x,y)=(0,0)$ yields two eigenvalues,
\begin{eqnarray}
\Lambda_{1}&=&\frac{1}{2(1-\sigma)} \, \left(r_1+r_2-2\sigma +
\sqrt {4\,{\sigma}^{2}+ (r_1-r_2)^2 }\right) 
, \\
\Lambda_{2}&=& \frac{1}{2(1-\sigma)} \, \left(r_1+r_2-2\sigma -
\sqrt {4\,{\sigma}^{2}+ (r_1-r_2)^2 }\right)  .
\end{eqnarray}
The eigenvalues for this system are of a similar form as those for linear migration in Eq.~(\ref{norm_eigs}), but multiplied by $1/(1-\sigma)$. Therefore, the mass action mixing can change the stability of the DFE as a function of basic reproduction numbers in the same settings as linear mixing. 

\section{Vaccination responses} \label{Sec:vaccination}

This section directly considers how the migration rates change the vaccination levels necessary to keep the DFE stable, which implies the occurrence of herd immunity. It explores whether neglecting the short- and long-term migration rates overestimates or underestimates the minimum vaccination rates necessary for disease fade-out.

We first restrict our attention to long-term migration only; i.e., letting $c_3 = 0$ in Eq.~(\ref{eq:SIR_coupled}). For $c_1>0$, this is equivalent to identifying the bifurcation points in $(v_1,v_2)$ when $\lambda_4=0$  in Eq.~(\ref{lambda4}).  Notice that if $c_1=0$, the subpopulations are isolated. The vaccination levels needed in each for the disease to die out is equivalent to solving for $v_1$ and $v_2$ in $\hat{R}_{1,2} \le 1$ from Eq.~(\ref{Eq:basicrep}). The constant solutions for $\hat{R}_{1,2} = 1$ using parameter values in Table~\ref{table:full_params} are shown in Figure~\ref{fig:c1vacc} as solid black lines, and the disease will die out in both populations in the top right quadrant. 

		\begin{figure}
		\begin{centering}
		\includegraphics[height=70mm]{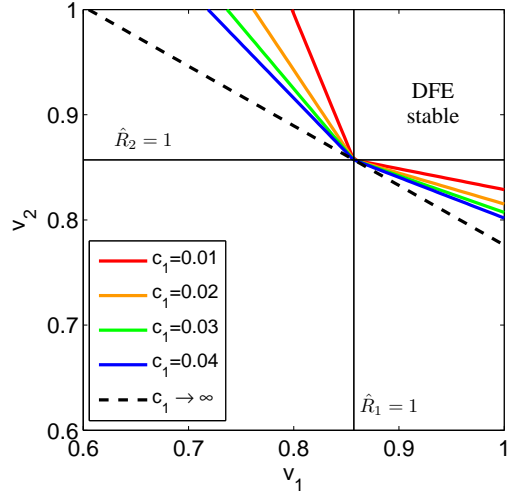}
		\caption{ {\label{fig:c1vacc} Boundary of the region of stability for the DFE as we vary $c_1$. The curves represent $\lambda_4=0$ using the parameter values in Table~\ref{table:full_params} and $c_3=0$.The vertical and horizontal black lines represent the vaccination rates necessary for a stable DFE in isolated populations ($c_1=0$). The dashed line represents the limiting curve as we increase $c_1$. }}
		\end{centering}
		\end{figure}

As we increase $c_1$, the boundary for the region of stability for the DFE spreads away from the $c_1=0$ case, increasing the die out region. The limit is a line, shown by the dotted black line in Figure~\ref{fig:c1vacc}. 
\begin{proposition}As $c_1\rightarrow \infty$, the bifurcation curve bounding the stable region approaches the line 
\begin{equation}\label{eq:vline}
v_2 =-\left(\frac{\mu_1}{\mu_2 \rho }\right)v_1+
\frac{(\mu_1+\mu_2 \rho) (\beta-\kappa)}{\beta \mu_2 \rho} -
\frac{(\mu_1+\mu_2 \rho)^2}{\beta \mu_2 \rho (1+\rho)} .
\end{equation}
\end{proposition}
This line is decreasing in $v_1$, with the slope depending on a ratio of the birth rates and subpopulation sizes. Note that if we use the basic reproductive numbers for the isolated subpopulations, Eq.~(\ref{eq:vline}) is equivalent to
\begin{equation}\label{eq:Rvline}
v_2 =-\left(\frac{\mu_1}{\mu_2 \rho }\right)v_1+
\frac{\left(1+\frac{\mu_1}{\mu_2 \rho}\right)}
{  (1+\rho)} \left(\left(\frac{\hat{R}_1(0)-1}{\hat{R}_1(0)}\right)+\left(\frac{\hat{R}_2(0)-1}{\hat{R}_2(0)}\right)\right).
\end{equation}
As $\hat{R}_1(0)$ and/or $\hat{R}_2(0)$ increase, the $v_2$ intercept increases and shifts the line up vertically. Therefore, the attainable stable DFE region in $(v_1,v_2)$ space decreases, as expected.

Next, we consider only short-term migration, i.e. letting $c_1 = 0$. Similarly, for $0 \le c_3 \le 1$, this is equivalent to identifying the bifurcation points in $(v_1,v_2)$ when $\lambda_4=0$. We fix the parameters to the values in Table~\ref{table:full_params} and vary $c_3$ to see the changes to the boundary of the stable region. The solution for $c_3=0$ is shown as solid black lines in Figure~\ref{fig:c3vac}. Again the disease will die out in both populations in the top right quadrant. As we increase $c_3$, the boundary for the region of stability for the DFE smoothly pulls away from top right quadrant, increasing the die out region. 
\begin{proposition}The limit as we increase $c_3 \rightarrow 1$ is \[ v_2=1-{\frac { \left( \kappa+\mu_{{2}} \right)  \left( \kappa+\mu_{{1}}
 \right) }{ \left( 1-v_{{1}} \right) {\beta}^{2}}} . \] 
\end{proposition}
Therefore, in both cases, underestimating the migration between populations causes an overestimation of the vaccination levels needed for herd immunity.

		\begin{figure}
		\begin{centering}
		\includegraphics[height=70mm]{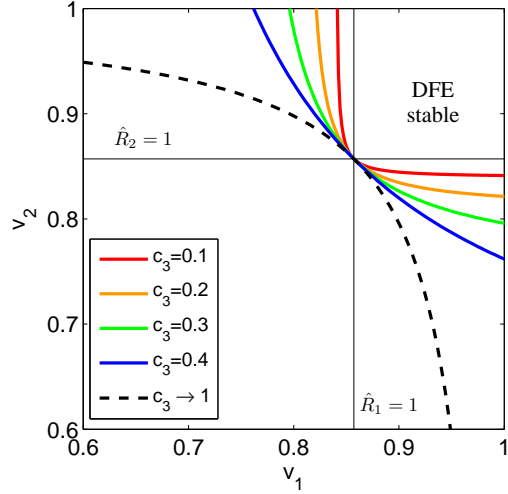}
		\caption{ {\label{fig:c3vac} Boundary of the region of stability for the DFE as we vary $c_3$. The curves represent $\lambda_4=0$ using the  parameter values in Table~\ref{table:full_params} and  $c_1=0$. The vertical and horizontal black lines represent the vaccination rates necessary for a stable DFE in isolated populations ($c_3=0$).}}
		\end{centering}
		\end{figure}

\section{Conclusions}\label{Sec:conclusion}

In this paper, we consider the effects of short- and long-term migration in coupled population models in the presence of vaccination. We study the interplay between the independent vaccination and migration rates across different populations.
We conclude that neglecting migration effects overestimates the vaccination levels necessary to achieve herd immunity. 

We have proven that if two isolated populations support an endemic state simultaneously, migration cannot change the stability of those endemic states. Analogously, this is also true for two populations with stable disease free equilibria. In contrast, migration can lead to disease die out in the mixed case. If a single population has a vaccination rate sufficient for herd immunity in isolation, low levels of migration from a population that is endemic will not necessarily make the disease endemic in both. In fact, increased levels of migration can lead to disease die out in both populations. However, migration rates are only physically realistic when they are small. 

Our results suggest more efficient vaccination strategies may be identified for groups of countries with significant migration between them. For example, instead of increasing the vaccination levels in a population that has already achieved herd immunity, sending vaccine to the less vaccinated neighboring country could have a greater impact on outbreak levels. The most efficient control algorithm would be to target the stable die out region as shown in Figure \ref{fig:c1vacc} or Figure \ref{fig:c3vac}. 

Conversely, populations for which vaccine delivery is difficult may benefit to a degree by vaccination of neighboring countries. More specifically, consider decreasing the migration rates to a country with a lower vaccination rate. We show in Figure \ref{fig:newvac} a policy where $N_2$, which has a vaccination rate $v_2=0.9$, decreases the long-term migration rate $c_1$ with $N_1$, which has a vaccination rate $v_1=0.7$. The short term migration rate is held constant at $c_3=0.1$. The decrease in number of new infections for $N_2$ is a small percentage of the increase in infections in $N_1$, and we conclude that the policy meant to help $N_2$ has unintended negative consequences for $N_1$. 

		\begin{figure}
		\begin{centering}
		\includegraphics[height=70mm]{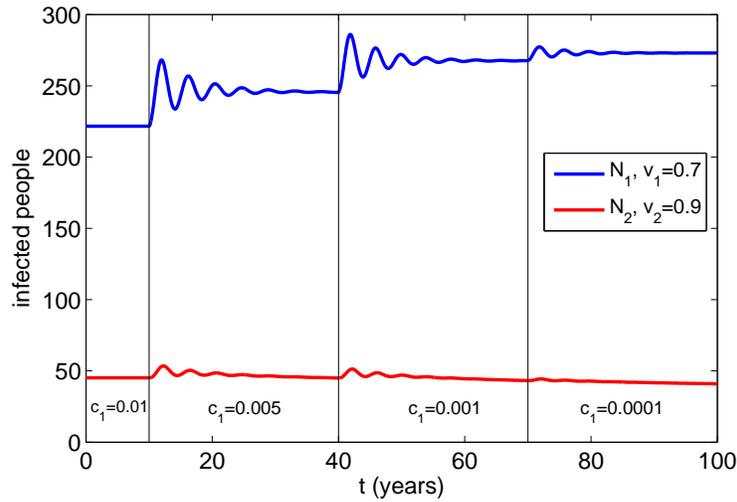}
		\caption{ \label{fig:newvac} Time series of infectives in both populations using the  parameter values in Table~\ref{table:full_params}, with $v_1=0.7$, $v_2=0.9$, and $c_3=0.1$. The value of $c_1$ is decreased to the constant noted in each window.}
		\end{centering}
		\end{figure}

In future directions, the model can be extended to include the effects of seasonality. A similar analysis of stable periodic behavior can reveal the sensitivity of synchronization to short-term and long-term migration. For example, the work of Schwartz \cite{Schwartz92} predicts new periodic orbits that can be excited by the mass action coupling in models with seasonal forcing. Specifically, these orbits exhibit long period outbreaks in small populations due to mass action coupling. When applying time dependent vaccination schedules, other parameters must be considered in addition to the average vaccination rates, such as pulse frequency and phase with respect to periodic application. Other techniques can be extended to migration models with vaccine control, such as prediction of future outbreaks as reported in Schwartz, et al.~\cite{Schwartz04}.

\section*{Acknowledgments}
We gratefully acknowledge support from the Office of Naval Research.  The authors were also supported by the National Institute of General Medical Sciences (Award No. R01GM090204). DATC holds a Career Award at the Scientific Interface from the Burroughs Wellcome Fund and received funding from the Bill and Melinda Gates Foundation Vaccine Modeling Initiative.  The content is solely the responsibility of the authors and does not necessarily represent the official views of the National Institute of General Medical Sciences or the National Institutes of Health. We would also like to thank Leah Shaw and Luis Mier-Y-Teran for their useful discussions. 


\end{document}